\documentclass[leqno]{article}
\usepackage{amsthm}
\usepackage{amsfonts}
\usepackage{amssymb}
\usepackage{amsmath}
\usepackage{mathtools}
\usepackage{hyperref}
\usepackage[capitalise]{cleveref}

\newtheorem{theorem}{Theorem}[section]
\newtheorem{lemma}[theorem]{Lemma}
\newtheorem{proposition}[theorem]{Proposition}
\newtheorem{corollary}[theorem]{Corollary}
\theoremstyle{definition}
\newtheorem{definition}[theorem]{Definition}

\newcommand{\A}{\mathcal A}
\newcommand{\B}{\mathcal B}
\newcommand{\M}{\mathcal M}
\newcommand{\CC}{\mathbb C}
\newcommand{\PP}{\mathbb P}
\newcommand{\ZZ}{\mathbb Z}
\newcommand{\EV}{\quad \Longleftrightarrow \quad}
\newcommand{\iso}{\cong}
\renewcommand{\:}{\colon}
\newcommand{\id}{\mathrm{id}}
\newcommand{\tr}{\mathrm{tr}}
\newcommand{\tensor}{\mathbin{\otimes}}
\newcommand{\<}{\langle}
\renewcommand{\>}{\rangle}

\allowdisplaybreaks

\title{Localizing quantum information}
\author{Fahimeh Bayeh and Andre Kornell}
\date{}

\newcommand{\Addresses}{{
  \bigskip
  \footnotesize

  \textsc{Department of Mathematics and Statistics, Dalhousie University}
  \par\nopagebreak
  \textsc{Halifax, NS B3H 4R2, Canada}
  \par\nopagebreak
  \textit{E-mail address}: \texttt{fahimeh.bayeh@dal.ca}

  \medskip

  \textsc{Department of Mathematical Sciences, New Mexico State University}
  \par\nopagebreak
  \textsc{Las Cruces, NM 88003, United States of America}
  \par\nopagebreak
  \textit{E-mail address}: \texttt{kornell@nmsu.edu}
}}

\begin{document}

\maketitle

\begin{abstract}
The idea that information can be localized is pervasive in physics. In thermodynamics, entropy flows from one reservoir to another, and in quantum gravity, the entropy content of some spatial regions is bounded. In quantum information theory, information is transmitted from one place to another, and entanglement can provide an advantage when there are constraints on such transmission.

Shannon entropy quantifies localizable information in the sense that marginalization defines an outer measure on the set of parts of a classical multipartite system. In contrast, von Neumann entropy does not quantify localizable information in this sense. However, a variant of von Neumann entropy, which originates in noncommutative geometry, does.

We prove that this adjusted von Neumann entropy is the minimum quantum entropy that quantifies localizable information. We also prove that this is the unique quantum entropy that characterizes Bell states in terms of redundant information, directly generalizing the classical case. We work with finite-dimensional $C^*$-algebras throughout, modeling finite quantum systems that may have superselection sectors.

\end{abstract}

\vspace{.5cm}


\vspace{.5cm}


\section{Introduction}

\subsection{Measures of information entropy}

The topic of this article is the quantum generalization of information entropy. For brevity and clarity, we interpret information entropy to be the quantity of information that the observer lacks about the microstate of a physical system, without endorsing this interpretation in physics. For simplicity, we consider only finite physical systems. Witten's review article \cite{Witten2020} is our basic reference.

When the information that the observer has about the microstate of a fully classical system is given by a probability distribution $p \in \CC^n$ on the set of microstates, the observer's information entropy is given by the familiar formula
\begin{equation}\label{equation 1}
S(p) = -\sum_{i=1}^n p_i \log p_i.
\end{equation}
We refer to the probability distribution as the state of the system and to this information entropy as the Gibbs entropy \cite{Gibbs1902} or \emph{Shannon entropy} \cite{Shannon1948} of the state.

The standard quantum analogue of classical information entropy was given by von Neumann. When the information of the observer about the microstate of a fully quantum system is given by a density matrix $\rho \in M_n(\CC)$, the observer's information entropy is given by the formula
\begin{equation}\label{equation 2}
S(\rho) = -\sum_{i=1}^n \lambda_i \log \lambda_i,
\end{equation}
where $\{\lambda_1, \ldots, \lambda_n\}$ is the spectrum of $\rho$ with multiplicity. We refer to the density matrix as the state of the system and to this information entropy as the \emph{von Neumann entropy} of the state \cite{vonNeumann1932}.

Shannon entropy and von Neumann entropy are remarkably similar in many formal and practical respects. Perhaps the starkest difference between them is that Shannon entropy is localizable but von Neumann entropy is not. Informally, we mean that the Shannon entropy of a bipartite state quantifies information that can be viewed as distributed between the two parts of the fully classical system, but von Neumann entropy cannot be interpreted in this way. Formally, we mean that Shannon entropy is monotonic, but von Neumann entropy is not.

Let $\A$ and $\B$ be two parts of a bipartite fully classical system in state $p$. Shannon entropy is \emph{monotonic} \cite[eq.~2.15]{Witten2020} in the sense that
\begin{equation*}
S(p_\A) \leq S(p),
\end{equation*}
where $p_\A$ is the marginal state of $\A$. Evidently, $S(p_\A)$ quantifies the information that the observer lacks about $\A$, and $S(p_\A) \leq S(p)$ because the observer lacks at least as much information about the bipartite system as a whole. Less precisely, the bipartite system contains at least as much information as the subsystem.

For contrast, consider the standard Bell scenario \cite{Bell1964}, in which a bipartite fully quantum system is in the state
$$
\rho = \frac 1 2 (|00\> + |11\>)(\<00| + \<11|) \in M_2(\CC) \tensor M_2(\CC) = M_4(\CC).
$$
In this case, $S(\rho) = 0$, but $S(\rho_\A) = \log 2$, so $S(\rho_\A) \not \leq S(\rho)$; see also \cite[eq.~3.43]{Witten2020}. Thus, von Neumann entropy does not quantify information that can be viewed as distributed between the two parts of this fully quantum system. There would be more information in a part of the system than in the system as a whole!

We observe that despite the striking similarity between equations \eqref{equation 1} and \eqref{equation 2}, von Neumann entropy is not the clear quantum analogue of Shannon entropy. Instead, this article investigates \emph{adjusted von Neumann entropy}
$$
\tilde S(\rho) = S(\rho) + \log n,
$$
where $\rho \in M_n(\CC)$ \cite{Kornell2025}. We show that adjusted von Neumann entropy is
\begin{enumerate}
\item the smallest quantum entropy such that $\tilde S(\rho_\A) \leq \tilde S(\rho)$ for all bipartite $\rho$,
\item the unique quantum entropy such that $\tilde S(\rho_\A) = \tilde S(\rho)$ for all Bell $\rho$.
\end{enumerate}

Before we gloss these results, we address the fact that adjusted von Neumann entropy $\tilde S$ differs from von Neumann entropy only by an additive constant. Thus, adjusted von Neumann entropy differs meaningfully from von Neumann entropy only when comparing states on fully quantum systems of different dimensions. This occurs naturally when we attempt to localize the information content of a quantum system, as we have done in the preceding discussion. This also occurs naturally when we consider quantum systems with superselection sectors of variable dimension. For this reason, adjusted von Neumann entropy is more properly defined in the setting of finite-dimensional $C^*$-algebras.

\subsection{Finite physical systems}

Finite systems are typically modeled by finite-dimensional $C^*$-algebras. In the case of a fully classical system with exactly $n$ pure states, that $C^*$-algebra is isomorphic to $\CC^n$. In the case of a fully quantum system with $n$ states in each maximal orthogonal set of pure states, that $C^*$-algebra is isomorphic to $\M_n \coloneqq M_n(\CC)$. In general, a finite-dimensional $C^*$-algebra $\A$ is of the form
$$
\A \iso \M_{n_1} \oplus \cdots \oplus \M_{n_s},
$$
and the factors $\M_{n_i}$, for $1 \leq i \leq s$, correspond to the superselection sectors of the system.

A state of such a system is modeled by a positive functional $\omega\: \A \to \CC$ with $\omega(1) = 1$. Such functionals are in one-to-one correspondence with density operators $\rho \in \A$, i.e., positive operators such that $\tr(\rho) = 1$. Density operators are widely known in quantum information theory, so we define a \emph{state} in $\A$ to be a density operator $\rho \in \A$.

The \emph{adjusted von Neumann entropy} of a state $\rho \in \A$ is defined to be
$$
\tilde S(\A, \rho) = - \tr(\rho \log \rho) + \tr (\rho \log \zeta_\A),
$$
where
$$
\zeta_\A = n_1 1_{n_1} \oplus \cdots \oplus n_s 1_{n_s}
$$
and $1_n$ denotes the identity operator in $\M_n$ for each positive integer $n$. The trace of an operator $a \in \A$ is its trace in $\M_n$, where $n = n_1 + \cdots + n_s$. Thus, when $a_i \in \M_{n_i}$ for each $1 \leq i \leq s$,
$$
\tr(a_1 \oplus \cdots \oplus a_s) = \tr(a_1) + \cdots + \tr(a_s).
$$
The trace in this paper is not normalized, so $\tr(1_n) = n$.

The adjusted von Neumann entropy is a \emph{quantum entropy} in the sense that
\begin{enumerate}
\item $\tilde S(\A, \rho) = \tilde S (\B, \sigma)$ if there is a $\dagger$-isomorphism $\Phi\: \A \to \B$ with $\Phi(\rho) = \sigma$,
\item $\tilde S(\A \oplus \B, \rho \oplus 0) = \tilde S(\A, \rho)$,
\item $\tilde S (\A \tensor \B, \rho \tensor \sigma) = \tilde S (\A, \rho) + \tilde S(\B, \sigma)$,
\item for all probabilistic mixtures $\rho = p_1 \rho_1 + p_2 \rho_2$ with $\rho_1$ orthogonal to $\rho_2$,
$$\tilde S(\A, \rho) = p_1 \tilde S(\A, \rho_1) + p_2 \tilde S(\A, \rho_2) + S(p);$$
\end{enumerate}
this is \cref{3.10}. The first axiom expresses that entropy is an isomorphism invariant of states. The second axiom expresses that impossible alternatives do not affect entropy. The third axiom expresses that entropy is additive for uncorrelated bipartite states. The fourth axiom is the familiar grouping axiom, which expresses that the entropy of a mixture of distinguishable states is equal to the entropy of the selection plus the expected entropy of the selected state. Note that Shannon entropy and von Neumann entropy both satisfy these axioms, albeit in the setting of commutative and simple finite-dimensional $C^*$-algebras, respectively.

The third axiom expresses additivity because the tensor product of finite-dimensional $C^*$-algebras models the composition of spatially separated systems. This tensor product may be constructed in a number of equivalent ways, e.g., as the tensor product of algebras over $\CC$. Thus, the state of a bipartite system is a state $\rho \in \A \tensor \B$, and its marginal states are $\rho_\A \coloneqq \tr_\B(\rho) \coloneqq (\id \tensor \tr)(\rho) \in \A$ and $\rho_\B \coloneqq \tr_\A(\rho) \coloneqq (\tr \tensor \id)(\rho) \in \B$. We can now gloss our first main result in the following way.

\begin{theorem}[also~\cref{3.19}]\label{1.1}
Adjusted von Neumann entropy $\tilde S$ is the minimum quantum entropy that is monotonic in the sense that $\tilde S(\rho_\A) \leq \tilde S(\rho)$ for all states $\rho \in \A \tensor \B$.
\end{theorem}

The monotonicity of a quantum entropy can equivalently be expressed in terms of the marginal state $\rho_\B$ as a consequence of the first axiom.

\subsection{Bell states}

The most important property of adjusted von Neumann entropy, which characterizes it among other quantum entropies, is its account of shared information in Bell scenarios. In the strictest sense, the Bell state is the two-qubit state
$$
\beta_2 = \frac 1 2 \sum_{x,y \in \{0,1\}} |x x\>\<y y| \in \M_2 \tensor \M_2.
$$
This state has a number of well-known generalizations, e.g., to the setting of two qudits. In the present article, we introduce and examine a basis-independent generalization of this notion that distills a key operational feature of the original Bell state.

The argument of Bell \cite{Bell1964} and the argument of Einstein, Podolsky, and Rosen that preceded it \cite{EinsteinPodolskyRosen1935} leverage the fact that, despite their randomness, measurement outcomes in one part of the bipartite system $\M_2 \tensor \M_2$ can be determined using an appropriate measurement in the other part. If, following consensus, we reject the transfer of information from one part of the system to the other as a result of these measurements, then we instead accept that whatever information provides the measurement outcomes in one part of the system also exists in the other part. Adjusted von Neumann entropy quantifies exactly this information.

We define a \emph{Bell state} to be a bipartite state $\rho \in \A \tensor \B$ such that every Boolean observable in $\A$ is perfectly correlated with some Boolean observable in $\B$, and vice versa. A \emph{Boolean observable} on a system $\A$ is a $\{0,1\}$-valued observable on $\A$, i.e., a self-adjoint operator $P \in \A$ whose spectrum is in $\{0, 1\}$. When $\A$ and $\B$ are both qubits, then a bipartite state $\rho \in \A \tensor \B$ is Bell iff it is equivalent to the Bell state $\beta_2$; see \cref{2.7}. We can now gloss our second main result in the following way.

\begin{theorem}[also~\cref{3.18}]\label{1.2}
Adjusted von Neumann entropy $\tilde S$ is the unique quantum entropy such that, for all Bell states $\rho \in \A \tensor \B$,
$$\tilde S(\rho_\A) = \tilde S(\rho) = \tilde S (\rho_\B).$$
\end{theorem}

Thus, adjusted von Neumann entropy is the unique quantum entropy such that a bipartite system in a Bell state contains as much information as each of its parts. Furthermore, it characterizes Bell states by this property.

\begin{theorem}[also~\cref{3.9}]\label{1.3}
A state $\rho \in \A \tensor \B$ is a Bell state iff $$\tilde S(\rho_\A) = \tilde S(\rho) = \tilde S (\rho_\B).$$
\end{theorem}

This equivalence is a direct generalization of the classical case. Furthermore, we contend that this equivalence is necessary for any coherent account of localizable information. The ability to resolve any proposition in $\A$ by resolving some proposition in $\B$ and vice versa is the availability of the same information in $\A$ as in $\B$. Of course, $\A$ and $\B$ contain the same information iff $\A \tensor \B$ contains the same amount of information as $\A$ and as $\B$.

Our discussion of \cref{1.3} suggests the following equivalence, which serves as a test for the coherence of this conceptual framework.

\begin{theorem}[also~\cref{3.8}]\label{1.4}
A state $\rho \in \A \tensor \B$ is a right-Bell state iff $$\tilde S(\rho) = \tilde S (\rho_\B).$$
\end{theorem}

\noindent Evidently, a state $\rho \in \A \tensor \B$ is \emph{right-Bell} if each Boolean observable in $\A$ is perfectly correlated with some Boolean observable in $\B$, and \cref{1.4} expresses that this is the case iff all information in $\A \tensor \B$ is contained in $\B$, which is the right part of the system. Our interpretation of the adjusted von Neumann entropy as quantifying localizable information led us to the statement of \cref{1.4}.

\subsection{Research context}

Right-Bell states are a special case of EPR states, which were introduced by Arens and Varadarajan \cite{ArensVaradarajan2000} in the context of type I factors and then immediately generalized to arbitrary von Neumann algebras by Werner \cite{Werner1999}. Explicitly, a right-Bell state is a bipartite state $\rho \in \A \tensor \B$ that is an EPR state relative to $(\A, \mathrm{Proj}(\A))$. In the present article, $\A$ and $\B$ are always finite-dimensional, though it is possible to extend the context to hereditarily atomic von Neumann algebras; see \cite[Remark~5.7]{Kornell2025}.

\cref{1.1,1.2} contribute to the axiomatic investigation of entropy. Shannon initiated this research program, proving \cite[Theorem~2]{Shannon1948}. Khinchin \cite{Khinchin1953} and Faddeev \cite{Faddeev1956} obtained similar characterizations. More recently, Baez, Fritz, and Leinster provided a category-theoretic axiomatization of Shannon entropy \cite[Theorem~2]{BaezFritzLeinster2011}.

Baez and Fritz then asked for a comparable characterization of von Neumann entropy based on the opposite of the category of finite-dimensional $C^*$-algebras and unital $\dagger$-homomorphisms \cite{Baez2011}. Such a characterization was obtained by Parzygnat \cite{Parzygnat2022}. Nakahira then obtained an elegant axiomatization of von Neumann entropy in terms of the category of finite-dimensional $C^*$-algebras and trace-preserving completely positive maps \cite[Theorem~3]{Nakahira2023}.

The question of Baez and Fritz was also a motivation for the research in the present article. Unital $\dagger$-homomorphisms can be characterized by the completely entropy-nonincreasing pullback of states for adjusted von Neumann entropy \cite[Theorem~1.1]{Kornell2025}, which suggests that the direct analogue of \cite[Theorem~2]{BaezFritzLeinster2011} should be an axiomatization of adjusted von Neumann entropy rather than of von Neumann entropy. Such a result can be extracted from section~\ref{section 3}, but we do not yet have a category-theoretic axiomatization of adjusted von Neumann entropy that matches \cite{BaezFritzLeinster2011} and \cite{Nakahira2023} in simplicity. Instead, \cref{1.2} characterizes adjusted von Neumann entropy in terms of Bell states.

In the same way, \cref{1.3} characterizes Bell states in terms of adjusted von Neumann entropy. It generalizes a combination of two well-known facts that are usually considered separately: first, a pure state $\rho \in \A \tensor \B$ with $\A = \M_n$ and $\B = \M_m$ is Bell iff it is maximally entangled with $n = m$ \cite{Schrodinger1935} and, second, $\rho$ is maximally entangled with $n = m$ iff
$$
S(\rho_\A) = \log m \qquad \text{ and } \qquad S(\rho_\B) = \log n
$$
\cite{NielsenChuang2010}. Thus, \cref{1.3} can be restated to characterize Bell states in terms of von Neumann entropy. Explicitly, a state $\rho \in \A \tensor \B$ is Bell iff
$$
S(\rho_\A) = S(\rho) + \tr(\rho_\B \log \zeta_\B) \qquad \text{ and } \qquad S(\rho_\B) = S(\rho) + \tr(\rho_\A \log \zeta_\A).
$$
The actual statement of \cref{1.3} might be misunderstood as tidying this equivalence with some bookkeeping. Its virtue is not that it is tidy but that it directly generalizes a conceptually rich equivalence in the classical case.

More generally, the significance of adjusted von Neumann entropy is that it generalizes Shannon entropy more faithfully than von Neumann entropy does. In his review article \cite{Witten2020}, Witten highlights the many qualitative differences between Shannon entropy and von Neumann entropy and comments on the ``miracle'' that von Neumann entropy is strongly subadditive despite these many differences. However, it then follows immediately that adjusted von Neumann entropy is strongly subadditive, so adjusted von Neumann entropy satisfies the same linear inequalities as Shannon entropy in tripartite systems \cite[Remark~1]{Han1981}. Thus, the miracle of strong subadditivity for von Neumann entropy can be explained as a remnant of the many features of adjusted von Neumann entropy that the latter shares with Shannon entropy.

This coincidence of properties between Shannon entropy and adjusted von Neumann entropy is not an accident. Indeed, adjusted von Neumann entropy first arose as a generalization of Shannon entropy in noncommutative geometry \cite{Kornell2025}. The ethos of this subject is a generalization of spaces to noncommutative algebras that minimally disturbs the properties of the spaces \cite{Gracia-BondiaVarillyFigueroa2001}. It was hence reasonable to expect that adjusted von Neumann entropy should share many features with Shannon entropy, and this expectation was a motivation for the research in the present article.

From the physical perspective, noncommutative geometry may be regarded as an interpretation of quantum theory. We briefly illustrate this thesis using the paradox of Einstein, Podolsky, and Rosen \cite{EinsteinPodolskyRosen1935}. In short, they argued that in a Bell scenario, ``no real change can take place in the second system in consequence of anything that may be done to the first system.'' In effect, noncommutative geometry embraces this position, implicitly glossing Bell scenarios in terms of hidden variables that range over quantum spaces.

We presume that this resolution would not have satisfied the authors of \cite{EinsteinPodolskyRosen1935}. It does not reveal an underlying classical theory that explains the strange phenomena of quantum mechanics, and it cannot \cite{Gleason1957, Bell1964, KochenSpecker1967, ClauserHorneShimonyHolt1969}. However, it does form the basis for a strikingly coherent account of these phenomena, of which \cref{1.4} is a small part. In this account, the information that determines the outcomes of local experiments exists locally; it is only of a more general kind than classical information. Adjusted von Neumann entropy quantifies this information.

The localization of information in quantum systems is particularly significant for quantum gravity. For example, relativistic constraints on the propagation of information indicate that entropy should be localizable, and the holographic principle assigns entropy bounds to regions of spacetime \cite{tHooft1993,Susskind1995,Bousso1999}. There is furthermore a superficial but intriguing resemblance between the expressions for adjusted von Neumann entropy and for generalized entropy \cite{Bekenstein1973}.

Adjusted von Neumann entropy is infinite in infinite-dimensional superselection sectors \cite[Remark~1.3]{Kornell2025}, but the Hilbert spaces that are associated to bounded regions of spacetime are often assumed to be finite-dimensional \cite{Bousso2002}. For example, the Hilbert space $\mathcal H$ associated to an ordinary causal diamond is sometimes explicitly taken to have dimension satisfying $\log(\dim \mathcal H) = a/4 = S$, where $a$ is the horizon area of the diamond, i.e., the area of its holographic screen, and $S$ is its maximum von Neumann entropy \cite[eqs.~1~and~2]{BanksDraperFarkas2021}. A recent derivation of the Einstein equation from the same connection between horizon area and entropy uses infinite-dimensional local algebras \cite{DorauMuch2026}, but these algebras are sometimes regarded as approximations \cite{Bousso2002}.

The holographic entropy bound is sometimes understood in terms of entanglement entropy \cite{BombelliKoulLeeSorkin1986,Srednicki1993,BoussoCasiniFisherMaldacena2014}, so von Neumann entropy is natural in this setting. We mention quantum gravity to suggest that adjusted von Neumann entropy may provide an alternative, more familiar, more classical account in this case, as it does in the case of Bell states.

\subsection*{Acknowledgment}

We thank Aaron David Fairbanks for contributing to the initial stage of this project.

\section{One-sided Bell states}

The two-fold purpose of this section is to review EPR states \cite{ArensVaradarajan2000}, introducing left-Bell and right-Bell states as a special case, and to prove \cref{2.5,2.6}, establishing a structural account for these one-sided Bell states in the setting of finite-dimensional $C^*$-algebras. We work with Boolean observables for convenience.

The \emph{Boolean observables} on $\A$ are the $\{0,1\}$-valued observables and, hence, are modeled by projections $P \in \A$. Given Boolean observables $P \in \A$ and $Q \in \B$, the probability of observing $1$ in both systems is the quantity
$$
\PP_\rho(P = 1, Q = 1) = \tr((P \tensor Q)\rho).
$$
Hence, the coincidence of measurement outcomes for Boolean observables $P$ and $Q$ can be expressed by the equation
$$\PP_\rho(P=0,Q =0) + \PP_\rho(P=1, Q =1) = 1.$$
Our first proposition examines this condition. Its content is folklore and also follows from results of Ozawa \cite{Ozawa2006}.

\begin{proposition}\label{2.1}
Let $\A$ and $\B$ be finite-dimensional $C^*$-algebras. Let $\rho \in \A \tensor \B$ be a state, and let $P \in \A$ and $Q \in \B$ be projections. Then, the following are equivalent:
\begin{enumerate}
\item $(P \tensor 1_\B)\rho = (1_\A \tensor Q)\rho$,
\item $\tr((P \tensor 1_\B)\rho) = \tr ((P \tensor Q) \rho) = \tr ((1_\A \tensor Q)\rho)$,
\item $\PP_\rho(P=0, Q=0) + \PP_\rho(P=1, Q=1) = 1$,
\item $\PP_\rho(P=0, Q=1) = 0$ and $\PP_\rho(P=1, Q = 0) = 0$,
\item $\PP_\rho(P=1) = \PP_\rho(Q=1) = 0$ or $\PP_\rho(P=1|Q=1) = \PP_\rho(Q=1|P=1) = 1$.
\end{enumerate}
\end{proposition}

\begin{proof}
First, we prove that $(1) \Leftrightarrow (2)$. The implication $(1) \Rightarrow (2)$ is trivial. We prove the implication $(2) \Rightarrow (1)$. Assume (2). It follows that $$\tr(\sqrt \rho (P \tensor (1_\B - Q))^2 \sqrt \rho) = \tr((P \tensor (1_\B-Q))\rho) = 0,$$
$$
\tr(\sqrt \rho ((1_\A - P) \tensor Q)^2 \sqrt \rho) = \tr(((1_\A -P) \tensor Q) \rho) = 0,
$$
so $(P \tensor (1_\B -Q)) \sqrt \rho = 0$ and $((1_\A - P) \tensor Q) \sqrt \rho = 0$ because $(X,Y) \mapsto \tr(X^\dagger Y)$ is an inner product on $\A \tensor \B$. Thus, $(P \tensor (1_\B -Q)) \rho = 0$ and $((1_\A - P) \tensor Q) \rho = 0$. Therefore, $(2) \Rightarrow (1)$.

Now, we prove the implication $(2) \Rightarrow (3)$. Assume (2). We calculate that
\begin{align*}
\PP_\rho(P=0&, Q=0) + \PP_\rho(P=1, Q = 1)
\\ &=
\PP_\rho(1_\A-P = 1, 1_\B-Q =1) + \PP_\rho(P=1,Q=1)
\\ &=
\tr(((1_\A-P) \tensor (1_\B-Q))\rho) + \tr((P \tensor Q)\rho)
\\ &
= \tr((1_\A \tensor 1_\B)\rho) - \tr((P\tensor 1_\B)\rho) - \tr((1_\A \tensor Q)\rho) + 2 \tr((P \tensor Q)\rho)
\\ & = \tr(\rho) + 0 = 1.
\end{align*}
Therefore, $(2) \Rightarrow (3)$.

The equivalence $(3) \Leftrightarrow (4)$ follows immediately from the equation
$$
\PP_\rho(P = 0, Q = 0) + \PP_\rho(P = 0, Q = 1) + \PP_\rho(P = 1, Q = 0) + \PP_\rho(P = 1, Q = 1) = 1.
$$

We prove the implication $(3) \Rightarrow (5)$. Assume (3). If $\PP_\rho(P=1) = 0$, then $\PP_\rho(P=1,Q=1) = 0$, so $\PP_\rho(P=0,Q=0) = 1$. In this case, $\PP_\rho(Q=0) =1$, and hence, $\PP_\rho(Q=1) = 0$, yielding (5). We have shown that $\PP_\rho(P=1) = 0$ implies $\PP_\rho(Q =1) = 0$; symmetrically, $\PP_\rho(Q =1) = 0$ implies $\PP_\rho(P=1) = 0$. If $\PP_\rho(P=1) \neq 0$, then $\PP_\rho(Q=1) \neq 0$. In this case, we have that
\begin{align*}
\PP_\rho(P=1|Q=1) & = \frac{\PP_\rho(P=1,Q=1)}{\PP_\rho(Q=1)} \\ & = \frac{\PP_\rho(Q=1) - \PP_\rho(P=0, Q=1)}{\PP_\rho(Q=1)} = \frac{\PP_\rho(Q=1)}{\PP_\rho(Q=1)} = 1
\end{align*}
because $\PP_\rho(P=0, Q=1) = 0$, and similarly, we have that $\PP_\rho(Q=1|P=1) = 1$, again yielding (5). Therefore, $(3) \Rightarrow (5)$.

We prove the implication $(5) \Rightarrow (2)$. Assume (5). If $\PP_\rho(P=1) = 0$ and $\PP_\rho(Q=1) = 0$, then $\PP_\rho(P=1, Q=1) = 0$, so $\tr((P \tensor 1_\B)\rho)$, $\tr((1_\A \tensor Q)\rho)$, and $\tr((P \tensor Q)\rho)$ are all equal to zero and, hence, to each other. Alternatively, if $\PP_\rho(P=1|Q=1) =1$ and $\PP_\rho(Q=1|P=1) =1$, then
$$
\frac{\tr((P\tensor Q)\rho)}{\tr((1_\A \tensor Q)\rho)} = \frac{\PP_\rho(P=1, Q=1)}{\PP_\rho(Q=1)} = 1,
$$
$$
\frac{\tr((P\tensor Q)\rho)}{\tr((P \tensor 1_\B)\rho)} = \frac{\PP_\rho(P=1, Q=1)}{\PP_\rho(P=1)} = 1,
$$
so $\tr((1_\A \tensor Q)\rho) = \tr((P\tensor Q)\rho) = \tr((P \tensor 1_\B)\rho)$. Therefore, $(5) \Rightarrow (2)$.
\end{proof}

We now use \cref{2.1} to define one-sided and two-sided Bell states.

\begin{definition}\label{2.2}
Let $\A$ and $\B$ be finite-dimensional $C^*$-algebras. Let $\rho \in \A \tensor \B$ be a state. We say that $\rho$ is
\begin{enumerate}
\item a \emph{right-Bell state} if, for each projection $P \in \A$, there exists a projection $Q \in \B$ such that any of the equivalent conditions in \cref{2.1} hold,
\item a \emph{left-Bell state} if, for each projection $Q \in \B$, there exists a projection $P \in \A$ such that any of the equivalent conditions in \cref{2.1} hold,
\item a \emph{Bell state} if it is both left-Bell and right-Bell.
\end{enumerate}
\end{definition}

Thus, a right-Bell state is a state with the property that the measurement outcomes of a Boolean observable on the left part of the system coincide with the measurement outcomes of some Boolean observable on the right part of the system; this is \cref{2.1}(3). Such a state is called a \emph{right}-Bell state because, intuitively, the right part of the system contains information about all measurement outcomes.

Right-Bell states and left-Bell states are natural classes of EPR states \cite{ArensVaradarajan2000, Werner1999}. Specifically, the right-Bell states are the EPR states relative to $(\A, \mathrm{Proj}(\A))$, and the left-Bell states are the EPR states relative to $(\B, \mathrm{Proj}(\B))$, where $\mathrm{Proj}(\A)$ is the set of projections in $\A$. They are also the EPR states relative to $(\A, \mathrm{Herm}(\A))$ and $(\B,\mathrm{Herm}(\B))$, respectively, because \cite{Werner1999} both $\mathrm{Proj}(\A)$ and $\mathrm{Herm}(\A)$ generate $\A$ as a von Neumann algebra. Thus, \cref{2.2} can be equivalently phrased in terms of self-adjoint operators, i.e., all observables.

We now show that a mixture of bipartite states is right-Bell only if the states themselves are right-Bell. This implies that every right-Bell state is a mixture of right-Bell pure states.

\begin{lemma}\label{2.3}
Let $\A$ and $\B$ be finite-dimensional $C^*$-algebras. Let $\rho_0, \rho_1 \in \A \tensor \B$ be states, let $t \in (0,1)$, and let $\rho = (1-t) \rho_0 + t \rho_1$. For all projections $P \in \A$ and $Q \in \B$, if $(P \tensor 1_\B) \rho = (1_\A \tensor Q)\rho$, then $(P \tensor 1_\B) \rho_0 = (1_\A \tensor Q)\rho_0$ and $(P \tensor 1_\B) \rho_1 = (1_\A \tensor Q)\rho_1$.
\end{lemma}

\begin{proof}
Assume that $(P \tensor 1_\B) \rho = (1_\A \tensor Q)\rho$. By \cref{2.1}, it follows that
$$(1-t)\PP_{\rho_0}(P = 0, Q = 1) + t \PP_{\rho_1}(P = 0, Q = 1) = \PP_\rho(P = 0, Q = 1) = 0$$ and that
$$(1-t)\PP_{\rho_0}(P = 1, Q = 0) + t \PP_{\rho_1}(P = 1, Q = 0) = \PP_\rho(P = 1, Q = 0) = 0.$$
Thus, $\PP_{\rho_0}(P = 0, Q = 1) = 0$, $\PP_{\rho_1}(P = 0, Q = 1) = 0$, $\PP_{\rho_0}(P = 1, Q = 0) = 0$, and $\PP_{\rho_1}(P = 1, Q = 0) = 0$. By \cref{2.1}, we conclude that $(P \tensor 1_\B) \rho_0 = (1_\A \tensor Q)\rho_0$ and $(P \tensor 1_\B) \rho_1 = (1_\A \tensor Q)\rho_1$.
\end{proof}

In the fully classical case, i.e., when $\A \iso \CC^n$ and $\B \iso \CC^m$, the state $\rho$ is essentially a probability distribution on the set $\{1, \ldots, n\} \times \{1, \ldots, m\}$. In this case, it is straightforward that $\rho$ is a right-Bell state iff that probability distribution is supported on the graph of a partial function $\{1, \ldots, m\} \to \{1, \ldots, n\}$. Similarly, $\rho$ is a Bell state iff that probability distribution is supported on the graph of a partial injection $\{1, \ldots, m\} \to \{1, \ldots, n\}$. Thus, in the classical case, all pure states on $\A \tensor \B$ are Bell.

The following theorem, which characterizes the pure states that are right-Bell in the fully quantum case, is a corollary of \cite[Theorem~5]{ArensVaradarajan2000}. We include a direct proof for simplicity and completeness. We write $E_{ij}$ for the matrix units.

\begin{theorem}\label{2.4}
Let $\A = \M_n$ and $\B = \M_m$, and let $\rho \in \A \tensor \B$ be a pure state. Then, the following are equivalent:
\begin{enumerate}
\item $\rho$ is a right-Bell state,
\item $n \leq m$, and there is a unitary $U\: \CC^m \rightarrow \CC^m$ such that
$$
(1_n \tensor U^\dagger) \rho (1_n \tensor U) = \frac 1 n \sum_{i, j = 1}^n E_{ij} \tensor E_{ij},
$$
\item $n \leq m$, and there is an isometry $V\: \CC^n \rightarrow \CC^m$ such that
$$
\rho = \frac 1 n \sum_{i, j = 1}^n E_{ij} \tensor V E_{ij} V^\dagger.
$$
\end{enumerate}
Furthermore, this isometry $V$ is unique up to a phase factor.
\end{theorem}

\begin{proof}
We prove $(1) \Rightarrow (3)$. Assume that $\rho$ is a right-Bell state. Since $\rho$ is pure, $\rho = |\psi\>\<\psi|$ for some unit vector $\psi \in \CC^n \tensor \CC^m$. This immediately implies that, for every projection $P$, there exists a projection $Q$ such that $$(P \tensor 1_\B)|\psi\> = (1_\A \tensor Q)|\psi\>.$$
Identifying $\CC^n$ with its dual space using the standard basis, we obtain a nonzero operator $V \: \CC^n \to \CC^m$ by $V = \sqrt n (\varepsilon \tensor 1_m) (1_n \tensor |\psi\>)$, where $\varepsilon\: \CC^n \tensor \CC^n \to \CC$ is the counit. We also obtain a transpose operation and, hence, a conjugation operation on $\A$. For each $A \in \A$, the conjugate $\overline A$ is just the entrywise conjugate of $A$ as a matrix.

Reasoning graphically \cite{Selinger2011}, we find that for every projection $P \in \A$, there exists a projection $Q \in \B$ such that $V \overline P = Q V$ and, hence, such that $V^\dagger V \overline P = V^\dagger Q V = \overline P V^\dagger V$. Thus, $V$ is a scalar multiple of an isometry. Composing with the unit $\eta\: \CC \to \CC^n \tensor \CC^n$, we find that
$$
|\psi\> = \frac 1 {\sqrt n} (1_n \tensor V) \eta = \frac 1 {\sqrt n} \sum_{i =1}^n |e_i\> \tensor V |e_i \>.
$$
Therefore, $V$ is an isometry, and
$$
\rho = |\psi\>\<\psi| = \frac 1 n \sum_{i, j = 1}^n E_{ij} \tensor V E_{ij} V^\dagger.
$$
If $V_1$ and $V_2$ are two isometries that yield $\rho$ in this way, then $V_1 E_{ij} V_1^\dagger = V_2 E_{ij} V_2^\dagger$ for all $1 \leq i, j \leq n$, so $V_1$ and $V_2$ are equal up to a phase factor.

To conclude $(3) \Rightarrow (2)$, we extend the isometry $V\: \CC^n \to \CC^m$ to a unitary $U\: \CC^m \to \CC^m$. To conclude $(2) \Rightarrow (1)$, it is enough to observe that the state
$$
\rho = \frac 1 n \sum_{i, j = 1}^n E_{ij} \tensor E_{ij}
$$
is right-Bell because it satisfies $(P \tensor 1_m) \rho = (1_n \tensor (\overline P \oplus 0_{m-n}))\rho$.
\end{proof}

A state $\rho \in \A \tensor \B$ is a mixture of finitely many pure states, each of which is supported on a single factor of $\A \tensor \B$. If $\rho$ is right-Bell, then so are these pure states by \cref{2.3}, and each of them is described by \cref{2.4}. However, not every mixture of finitely many right-Bell pure states is itself right-Bell. For example, it is easy to show that, when $\rho_1, \rho_2, \rho_3, \rho_4 \in \M_2 \tensor \M_2$ are the standard Bell states, then their uniform mixture $\rho \in \M_4$ is not right-Bell because it is the maximally mixed state. Our next theorem provides the condition that ensures that a mixture of right-Bell states is itself right-Bell. We say that states $\rho_1$ and $\rho_2$ are \emph{orthogonal} and write $\rho_1 \perp \rho_2$ if $\tr(\rho_1 \rho_2) = 0$ or, equivalently, $\rho_1\rho_2 = 0$.

\begin{theorem}\label{2.5}
Let $\A$ and $\B$ be finite-dimensional $C^*$-algebras. Let
$$\rho = p_1 \rho_1 + \cdots + p_k \rho_k$$
for pairwise-orthogonal right-Bell states $\rho_1, \ldots, \rho_k \in \A \tensor \B$ and positive real numbers $p_1, \ldots, p_k$ that sum to $1$. The following are equivalent:
\begin{enumerate}
\item $\rho$ is right-Bell,
\item the states $\tr_\A (\rho_i)$, for $1 \leq i \leq k$, are pairwise-orthogonal.
\end{enumerate}
\end{theorem}

\begin{proof}
We prove $(1) \Rightarrow (2)$. Assume that $\rho$ is right-Bell. By Lemma~\ref{2.3}, without loss of generality, $k = 2$, and $\rho_1$ and $\rho_2$ are right-Bell pure states. For $i \in \{1,2\}$, let $P_i$ be the unique minimal central projection of $\A$ such that $(P_i \tensor 1_\B)\rho_i = \rho_i$, and let $Q_i$ be the unique minimal central projection of $\B$ such that $(1_\A \tensor Q_i) \rho_i = \rho_i$.

Assume that $Q_1 \neq Q_2$. Then, we reason that for both $i \in \{1,2\}$,
$$
Q_i\tr_\A(\rho_i) = \tr_\A((1_\A \tensor Q_i)\rho_i) = \tr_\A(\rho_i).
$$
Therefore, $\tr_\A(\rho_1) \perp \tr_\A(\rho_2)$ because $Q_1 \perp Q_2$.

Assume that $P_1 \neq P_2$. Let $Q \in \B$ be a projection such that $(P_1 \tensor 1_\B) \rho = (1_\A \tensor Q)\rho$. We have that $(P_1 \tensor 1_\B) \rho_1 = \rho_1$, that $((1_\A-P_1) \tensor 1_\B) \rho_2 = \rho_2$, that $(P_1 \tensor 1_\B) \rho = (1_\A \tensor Q)\rho$, and that $((1_\A -P_1) \tensor 1_\B)\rho = (1_\A \tensor (1_\B - Q))\rho$. Appealing to Lemma~\ref{2.3}, we reason that
$$
Q \tr_\A (\rho_1) = \tr_\A ((1_\A \tensor Q)\rho_1) = \tr_\A((P_1 \tensor Q) \rho_1) = \tr_\A(\rho_1).
$$
Similarly, $(1_\B -Q)\tr_\A(\rho_2) = \tr_\A(\rho_2)$. Therefore, $\tr_\A(\rho_1) \perp \tr_\A(\rho_2)$ because $Q \perp 1_\B-Q$.

Assume that $P_1 = P_2$ and $Q_1 = Q_2$. We have that $(1_\A \tensor Q_1) \rho = \rho$. For all projections $P \in \A$, let $Q \in \B$ be a projection such that $(P \tensor 1_\B)\rho = (1_\A \tensor Q)\rho$. We reason that $ (P \tensor 1_\B)\rho = (1_\A \tensor Q)\rho = (1_\A \tensor Q) (1_\A \tensor Q_1)\rho = (1_\A \tensor Q Q_1) \rho. $ Hence, for all projections $P \in \A$ and, in particular, for all projections $P \leq P_1$, there exists a projection $Q \leq Q_1$ such that $(P \tensor 1_\B) \rho = (1_\A \tensor Q) \rho$. It follows that $\rho \in P_1 \A \tensor Q_1 \B$ is a right-Bell state. Of course, $P_1 \A \iso \M_n$ and $Q_1 \B \iso \M_m$ for some $m, n \geq 1$.

We assume without loss of generality that $\A = \M_n$ and $\B = \M_m$ and argue as in the proof of Theorem~\ref{2.4}. Let $\varepsilon \: \CC^n \tensor \CC^n \to \CC$ be the counit. For each projection $P \in \A$, there is a projection $Q \in \B$ such that $(P \tensor 1_\B) \rho = (1_\A \tensor Q)\rho$ and, by \cref{2.3}, such that $(P \tensor 1_\B) \rho_i = (1_\A \tensor Q)\rho_i$ for both $i \in \{1,2\}$. Writing $\rho_i = |\psi_i\> \<\psi_i|$, we infer that $(P \tensor 1_\B) |\psi_i\> = (1_\A \tensor Q) |\psi_i\>$ for both $i \in \{1,2\}$. Reasoning as in the proof of Theorem~\ref{2.4}, we infer that $V_i \overline P = Q V_i$ for both $i \in \{1, 2\}$, where $V_i = \sqrt n (\varepsilon \tensor 1_m) (1_n \tensor |\psi_i\>)$. Applying this conclusion, we find that for all projections $P \in \A$, there is a projection $Q \in \B$ such that $V_1^\dagger V_2 \overline P = V_1^\dagger Q V_2 = \overline P V_1^\dagger V_2$. Thus, $V_1^\dagger V_2$ is in the center of $\A$ and, hence, a scalar multiple of the identity $1_n \in \A$.

We now compute that
$$
0 = \tr(\rho_1 \rho_2) = | \< \psi_1 |\psi_2\>|^2 = \frac 1 {n^2} |\tr(V_1^\dagger V_2)|^2.
$$
Since $V_1^\dagger V_2$ is a scalar, we find that $V_1^\dagger V_2 = 0$. Reasoning graphically, it is easy to see that $\tr_\A(\rho_i) = \frac 1 n V_i V_i^\dagger$ for both $i \in \{1,2\}$. Therefore, $\tr_\A(\rho_1) \perp \tr_\A(\rho_2)$. Having reached this conclusion in all three cases, we have proved $(1) \Rightarrow (2)$.

We prove $(2) \Rightarrow (1)$. Assume that the states $\tr_\A(\rho_i)$, for $1 \leq i \leq k$, are pairwise-orthogonal, and let $P \in \A$ be a projection. For each index $i$, there exists a projection $Q_i \in \B$ such that $(P \tensor 1_\B)\rho_i = (1_\A \tensor Q_i)\rho_i$. We now construct a projection $Q \in \B$ such that $(P \tensor 1_\B)\rho_i = (1_\A \tensor Q)\rho_i$ for all $i$.

For each index $i$, let $R_i \in \B$ be the support projection of the state $\tr_\A(\rho_i)$. In other words, $R_i$ is the projection onto the range of $\tr_\A(\rho_i)$. For $i \neq j$, we calculate that
$$
\tr_\A((1_\A \tensor R_i)\rho_j) = R_i \tr_\A(\rho_j) = 0,
$$
which implies that $(1_\A \tensor R_i)\rho_j = 0$. Similarly, we calculate that
$$
\tr_\A((1_\A \tensor (1_\B - R_i))\rho_i) = (1_\B - R_i)\tr_\A(\rho_i) = 0,
$$
which implies that $(1_\A \tensor R_i)\rho_i = \rho_i$. It follows that
\begin{align*}&
(1_\A \tensor (Q_i \wedge R_i))\rho_i = \lim_{l \to \infty} (1_\A \tensor Q_i R_i)^l \rho_i = \lim_{l \to \infty} (P \tensor 1_\B) \rho_i = (P \tensor 1_\B) \rho_i.
\end{align*}
This computation uses a well-known formula of von Neumann for the greatest lower bound of two projections \cite[Lemma~22]{vonNeumann1949}.

The support projections $R_i$, for $1 \leq i \leq k$, are pairwise-orthogonal because the states $\tr_\A(\rho_i)$ are pairwise-orthogonal. It follows that $Q \coloneqq \sum_{i=1}^k Q_i \wedge R_i$ is a projection in $\B$. We now calculate that
\begin{align*}
(1_\A \tensor Q)\rho
&=
\sum_{i,j = 1}^k p_j (1_\A \tensor (Q_i \wedge R_i)) \rho_j
=
\sum_{i = 1}^k p_i (1_\A \tensor (Q_i \wedge R_i)) \rho_i
\\ & =
\sum_{i = 1}^k p_i (P \tensor 1_\B) \rho_i = (P \tensor 1_\B)\rho.
\end{align*}
Therefore, $\rho$ is a right-Bell state. We have proved $(2) \Rightarrow (1)$.
\end{proof}

We might paraphrase \cref{2.5} as saying that two orthogonal right-Bell states $\rho_1$ and $\rho_2$ mix into a right-Bell state iff they can be distinguished using a measurement on $\B$. The theorem fails if $\rho_1$ and $\rho_2$ are not orthogonal to each other. For example, if $\rho \in \M_2 \tensor \M_2$ is a Bell state, then the mixtures of $\rho_1 = \rho \tensor |\psi_1\>\<\psi_1|$ and $\rho_2 = \rho \tensor |\psi_2\>\<\psi_2|$ are right-Bell states for all unit vectors $\psi_1, \psi_2 \in \CC^2$.

The following corollary combines \cref{2.5} and \cref{2.4}, generalizing the latter to all states $\rho \in \A \tensor \B$ when $\A$ and $\B$ are simple. It is straightforward to obtain a further generalization to arbitrary finite-dimensional $C^*$-algebras $\A$ and $\B$, but the formal statement is more complicated than enlightening.

\begin{corollary}\label{2.6}
Let $\A = \M_n$, let $\B= \M_m$, and let $\rho \in \A \tensor \B$ be a state. Then, $\rho$ is right-Bell iff there exist
\begin{enumerate}
\item scalars $p_1, \ldots, p_k \in (0, 1]$,
\item isometries $V_1, \ldots, V_k \: \CC^n \to \CC^m$,
\end{enumerate}
such that $p_1 + \cdots + p_k = 1$, such that $V_r^\dagger V_s = 0$ whenever $r \neq s$, and such that
$$
\rho = \sum_{r=1}^k \sum_{i,j=1}^n \frac {p_r} n E_{ij} \tensor V_r E_{ij} V_r^\dagger.
$$
\end{corollary}

\begin{proof}
Assume that $\rho$ is right-Bell. Diagonalization immediately yields a decomposition $\rho= p_1 \rho_1 + \cdots + p_k \rho_k$ where $p_1, \ldots, p_k \in (0, 1]$ satisfy $p_1 + \cdots + p_k = 1$ and $\rho_1, \ldots, \rho_k \in \A \tensor \B$ are pairwise-orthogonal pure states. By Theorem~\ref{2.4}, each pure state $\rho_r$ is of the form $\rho_r = \frac 1 n \sum_{i,j=1}^n E_{ij} \tensor V_r E_{ij} V_r^\dagger$ for some isometry $V_r \: \CC^n \to \CC^m$, and by Theorem~\ref{2.5}, the states $\tr_\A(\rho_r) = \frac 1 n V_r V_r^\dagger$ are pairwise-orthogonal. We have proved the forward implication.

For the other implication, assume that there exist scalars $p_1, \ldots, p_k \in (0,1]$ and isometries $V_1, \ldots, V_k \: \CC^n \to \CC^m$ that satisfy the conclusion of the corollary. For each $1 \leq r \leq k$, let $\rho_r = \frac 1 n \sum_{i,j=1}^n E_{ij} \tensor V_r E_{ij} V_r^\dagger$. The pure states $\rho_r$ are right-Bell by Theorem~\ref{2.4} and pairwise-orthogonal because $V_r^\dagger V_s = 0$ whenever $r \neq s$. Furthermore, the states $\tr_\A(\rho_r) = \frac 1 n V_r V_r^\dagger$ are pairwise-orthogonal for the same reason. Theorem~\ref{2.5} now implies that the state $\rho$ is right-Bell. We have proved the backward implication.
\end{proof}

We now conclude this section with a variant of \cref{2.6} for Bell states. The structure of Bell pure states is well known; it is implied by \cref{2.4}. The following corollary observes that every Bell state is pure.

\begin{corollary}\label{2.7}
Let $\A = \M_n$, let $\B= \M_m$, and let $\rho \in \A \tensor \B$ be a state. Then, $\rho$ is Bell iff $n = m$ and there exists a unitary $U \in \M_n$ such that
$$
\rho = \frac 1 n \sum_{i,j=1}^n E_{ij} \tensor U E_{ij} U^\dagger.
$$
\end{corollary}

\begin{proof}
The backward direction is trivial because the conclusion of the corollary is equivalent to the equation
$$
(1_n \tensor U)^\dagger \rho (1_n \tensor U) = \frac 1 n \sum_{i,j=1}^n E_{ij} \tensor E_{ij},
$$
and the right side is the standard Bell state. For the forward direction, assume that $\rho$ is a Bell state. By Corollary~\ref{2.6}, there exist isometries $\CC^n \to \CC^m$ and $\CC^m \to \CC^n$, so $n = m$. Furthermore, by the same Corollary~\ref{2.6}, there exist scalars $p_1, \ldots, p_k \in (0, 1]$ and isometries $V_1, \ldots, V_k \: \CC^n \to \CC^n$ such that $p_1 + \cdots + p_k = 1$, such that $V_i^\dagger V_j = 0$ whenever $i \neq j$, and such that
$$
\rho = \sum_{r=1}^k \sum_{i,j=1}^n \frac {p_r} n E_{ij} \tensor V_r E_{ij} V_r^\dagger.
$$
The isometries $V_1, \ldots, V_k$ are unitaries, so $k=1$ and $p_1 = 1$. This proves the forward implication.
\end{proof}

\section{Entropy}\label{section 3}

This section establishes that both the Bell states $\rho \in \A \tensor \B$ and adjusted von Neumann entropy $\tilde S$ are characterized in terms of each other by the equations $$\tilde S (\tr_\A(\rho)) = \tilde S(\rho) = \tilde S(\tr_\B(\rho)).$$
We begin this section by reviewing the definitions of von Neumann entropy and adjusted von Neumann entropy. We follow the standard convention that $0\log 0 = 0$ throughout.

\begin{definition}[\cite{vonNeumann1932,Segal1960}]
Let $\A$ be a finite-dimensional $C^*$-algebra. Then, the \emph{von Neumann entropy} of a state $\rho \in \A$ is the real number
$$
S(\rho) = - \tr (\rho \log \rho).
$$
\end{definition}

Von Neumann entropy satisfies the following inequality for mixtures of states, which directly generalizes the same inequality for Shannon entropy.

\begin{theorem}[{\cite[Theorem 11.10]{NielsenChuang2010}}]\label{3.2}
Let $\A$ be a finite-dimensional $C^*$-algebra. Let $p_1, \ldots, p_k \in (0,1]$ sum to $1$, and let $\rho_1, \ldots, \rho_k \in \A$ be states. For the state $\rho = p_1 \rho_1 + \cdots + p_k \rho_k$,
$$
S(\rho) \leq \sum_{i=1}^k p_i S(\rho_i) - \sum_{i=1}^k p_i \log p_i.
$$
Furthermore, this inequality is an equality iff the states $\rho_i$, for $1 \leq i \leq k$, are pairwise-orthogonal.
\end{theorem}

Adjusted von Neumann entropy is defined in terms of the counting weight on a finite-dimensional $C^*$-algebra. Explicitly, the \emph{counting weight} on
$$\A = \M_{n_1} \oplus \cdots \oplus \M_{n_s}$$
is the operator
$$
\zeta_\A = n_1 1_{n_1} \oplus \cdots \oplus n_{s} 1_{n_s}.
$$

\begin{definition}[\cite{Segal1960,Kornell2025}]\label{3.3}
Let $\A$ be a finite-dimensional $C^*$-algebra. Then, the \emph{adjusted von Neumann entropy} of a state $\rho \in \A$ is
$$
\tilde S (\rho) = S(\rho) + \tr(\rho \log \zeta_\A).
$$
\end{definition}

We remark that adjusted von Neumann entropy can be expressed in terms of Umegaki relative entropy \cite[sec.~4]{Umegaki1962} via $\tilde S(\rho) = \log(\dim \A) - S(\rho || \zeta_\A /\dim \A)$ \cite[Remark~1.5]{Kornell2025}. The additive constant $\log(\dim \A)$ is significant for \cite[Theorem~1.1]{Kornell2025} and, similarly, for the main results of the present paper.

We now prove a sequence of propositions toward \cref{3.8}, which relates adjusted von Neumann entropy to right-Bell states. The following proposition expresses that discarding a part of a bipartite system can never increase the entropy of that system. The corresponding inequality for von Neumann entropy is undoubtedly less natural. In the case that $\A = \M_n$ and $\B = \M_m$, it follows from the Araki-Lieb inequality \cite[eq.~3.1]{ArakiLieb1970}.

\begin{proposition}\label{3.4}
Let $\A$ and $\B$ be finite-dimensional $C^*$-algebras, and let $\rho \in \A \tensor \B$ be a state. Then,
\begin{enumerate}
\item $\tilde S(\tr_\A(\rho)) \leq \tilde S(\rho)$,
\item $S(\tr_\A(\rho)) \leq S(\rho) + \tr(\rho(\log\zeta_\A \tensor 1_\B) )$.
\end{enumerate}
\end{proposition}

\begin{proof}
Claim 1 follows by \cite[Theorem~1.1]{Kornell2025} because the trace-preserving completely positive map $\tr_\A\: \A \tensor \B \to \B$ is adjoint to the inclusion unital $\dagger$-homomorphism $\B \to \A \tensor \B$. Claim 2 is equivalent to claim 1 because
\begin{align*}
\tilde S(\tr_\A(\rho))
& =
S(\tr_\A(\rho)) + \tr( \tr_\A(\rho) \log \zeta_\B)
\\ &=
S(\tr_\A(\rho)) + \tr( \tr_\A(\rho(1_\A \tensor \log \zeta_\B)))
\\ & =
S(\tr_\A(\rho)) + \tr (\rho (1_\A \tensor \log \zeta_\B)),
\end{align*}
\begin{align*}
\tilde S(\rho)
& =
S(\rho) + \tr(\rho \log \zeta_{\A \tensor \B})
=
S(\rho) + \tr(\rho \log (\zeta_\A \tensor \zeta_\B))
\\ & =
S(\rho) + \tr(\rho(\log \zeta_\A \tensor 1_\B) ) + \tr (\rho (1_\A \tensor \log \zeta_\B)).
\end{align*}
\end{proof}

We verify that \cref{3.2} holds for adjusted von Neumann entropy as well.

\begin{proposition}\label{3.5}
Let $\A$ be a finite-dimensional $C^*$-algebra. Let $p_1, \ldots, p_k \in (0,1]$ sum to $1$, and let $\rho_1, \ldots, \rho_k \in \A$ be states. For $\rho = p_1 \rho_1 + \cdots + p_k \rho_k$,
$$
\tilde S(\rho) \leq \sum_{i=1}^k p_i \tilde S(\rho_i) - \sum_{i=1}^k p_i \log p_i.
$$
Furthermore, this inequality is an equality iff the states $\rho_i$, for $1 \leq i \leq k$, are pairwise-orthogonal.
\end{proposition}

\begin{proof}
This proposition is an immediate corollary of Theorem~\ref{3.2} because
$$
\sum_{i = 1}^k p_i \tr(\rho_i \log \zeta_\A) = \tr(\rho \log \zeta_\A).
$$
\end{proof}

We now prove counterparts to Theorems~\ref{2.5} and \ref{2.4}, in that order.

\begin{lemma}[{cf.~\cite[Exercise~11.16]{NielsenChuang2010}}]\label{3.6}
Let $\A$ and $\B$ be finite-dimensional $C^*$-algebras. Let
$$
\rho = p_1 \rho_1 + \cdots + p_k \rho_k
$$
for pairwise-orthogonal states $\rho_1, \ldots, \rho_k \in \A \tensor \B$ and positive real numbers $p_1, \ldots, p_k$ that sum to $1$. The following are equivalent:
\begin{enumerate}
\item $\tilde S(\tr_\A(\rho)) = \tilde S(\rho)$,
\item the states $\tr_\A(\rho_i)$, for $1 \leq i \leq k$, are pairwise-orthogonal and satisfy $\tilde S(\tr_\A(\rho_i)) = \tilde S(\rho_i)$.
\end{enumerate}
\end{lemma}

\begin{proof}
Appealing to \cref{3.4} and Proposition~\ref{3.5}, we calculate that
\begin{align*}
\tilde S(\tr_\A(\rho))
& \leq
\sum_{i=1}^k p_i \tilde S (\tr_\A(\rho_i)) - \sum_{i=1}^k p_i \log p_i
\\ & \leq
\sum_{i=1}^k p_i \tilde S (\rho_i) - \sum_{i=1}^k p_i \log p_i = \tilde S(\rho).
\end{align*}
Thus, $\tilde S(\tr_\A(\rho)) = \tilde S(\rho)$ iff both inequalities in this calculation are equalities. By Proposition~\ref{3.5}, the first inequality is an equality iff the states $\tr_\A(\rho_i)$ are pairwise-orthogonal. By \cref{3.4}, the second inequality is an equality iff $\tilde S(\tr_\A(\rho_i)) = \tilde S(\rho_i)$ for all states $\tr_\A(\rho_i)$.
\end{proof}

\begin{proposition}\label{3.7}
Let $\A = \M_n$ and $\B = \M_m$, and let $\rho \in \A \tensor \B$ be a pure state. Then, the following are equivalent:
\begin{enumerate}
\item $\tilde S (\tr_\A(\rho)) = \tilde S (\rho)$,
\item $n \leq m$, and there is an isometry $V\: \CC^n \to \CC^m$ such that
$$
\rho = \frac 1 n \sum_{i, j = 1}^n E_{ij} \tensor V E_{ij} V^\dagger.
$$
\end{enumerate}
\end{proposition}

\begin{proof}
Because the state $\rho$ is pure, $\tilde S(\rho) = \log m + \log n$, and there exists an operator $V\: \CC^n \to \CC^m$, which is unique up to a phase factor, such that
$$
\rho = \frac 1 n \sum_{i,j = 1}^n E_{ij} \tensor V E_{ij} V^\dagger.
$$
The existence and uniqueness of the operator $V$ follows from the operator-vector correspondence \cite[sec.~1.1.2]{Watrous2018}, which is easily seen in the graphical calculus \cite{Selinger2011}. Furthermore,
$$
\tr_\A(\rho) = \frac 1 n \sum_{i = 1}^n V E_{ii} V^\dagger = \frac 1 n V V^\dagger.
$$
If $V$ is an isometry, then $V V^\dagger$ is a rank-$n$ projection, so $S(\tr_\A(\rho)) = \log n$ and $\tilde S(\tr_\A(\rho)) = \log n + \log m = \tilde S(\rho)$. Therefore, $(2) \Rightarrow (1)$. Conversely, if $\tilde S(\tr_\A(\rho)) = \tilde S (\rho)$, then $S(\tr_\A(\rho)) = \log n$. Since the state $\tr_\A(\rho)$ has rank no greater than $n$, this conclusion implies that $V V^\dagger$ is a rank-$n$ projection \cite[eq.~3.29]{Witten2020} and, hence, that $V$ is an isometry. Therefore, $(1) \Rightarrow (2)$.
\end{proof}

We finally assemble the pieces into a straightforward proof of one of our main results.

\begin{theorem}\label{3.8}
Let $\A$ and $\B$ be finite-dimensional $C^*$-algebras, and let $\rho \in \A \tensor \B$ be a state. Then, the following are equivalent:
\begin{enumerate}
\item $\tilde S (\tr_\A(\rho)) = \tilde S (\rho)$,
\item $\rho$ is right-Bell.
\end{enumerate}
\end{theorem}

\begin{proof}
Diagonalizing $\rho$, we obtain a decomposition $\rho = p_1 \rho_1 + \cdots + p_k \rho_k$, where $p_1, \ldots, p_k$ are positive real numbers that sum to $1$ and $\rho_1, \ldots, \rho_k$ are pairwise-orthogonal pure states. We now reason that
\begin{align*}
 \tilde S (\tr_\A(\rho)) = \tilde S (\rho) \EV &
\tr_\A(\rho_i) \perp \tr_\A(\rho_j) \text{ for all } i \neq j
\\ &
\text { and } \tilde S (\tr_\A(\rho_i)) = \tilde S (\rho_i) \text{ for all } i
\\ \EV & \tr_\A(\rho_i) \perp \tr_\A(\rho_j) \text{ for all } i \neq j
\\ &
\text { and } \rho_i \text{ is right-Bell for all } i
\\ \EV & \rho \text{ is right-Bell. }
\end{align*}
The first equivalence comes from \cref{3.6}. The third equivalence comes from \cref{2.3} and \cref{2.5}. It remains to justify the second equivalence.

Let $1 \leq i \leq k$. Let $P_i \in \A$ and $Q_i \in \B$ be the unique minimal central projections such that $\rho_i \in P_i \A \tensor Q_i \B$. We observe that $\rho_i$ is a right-Bell state in $\A \tensor \B$ iff it is a right-Bell state in $P_i \A \tensor Q_i \B$ because $(P \tensor 1_\B) \rho_i = (P P_i \tensor 1_\B) \rho_i$ and $(1_\A \tensor Q) \rho_i = (1_\A \tensor Q Q_i) \rho_i$ for all projections $P \in \A$ and $Q \in \B$. We also observe that the entropy $\tilde S(\rho_i)$ is the same in both $\A \tensor \B$ and $P_i \A \tensor Q_i \B$ and that, similarly, the entropy $\tilde S(\tr_\A(\rho_i))$ is the same in both $\B$ and $Q_i \B$. Since the factors $P_i\A$ and $Q_i\B$ are isomorphic to $\M_{n_i}$ and $\M_{m_i}$, respectively, for some positive integers $n_i$ and $m_i$, the second equivalence now follows from \cref{3.7} and \cref{2.4}. Having established all three equivalences, we have proved the theorem.
\end{proof}

In the special case that $\A = \M_n$ and $\B = \M_m$, \cref{3.8} can also be proved by combining \cref{2.6} with \cite[Theorem~1.3]{CarlenLieb2012}; the resulting proof is more opaque and not much shorter. \cref{3.8} has an obvious variant for Bell states, which we prove now.

\begin{corollary}\label{3.9}
Let $\A$ and $\B$ be finite-dimensional $C^*$-algebras, and let $\rho \in \A \tensor \B$ be a state. Then, the following are equivalent:
\begin{enumerate}
\item $\tilde S (\tr_\A(\rho)) = \tilde S (\rho) = \tilde S(\tr_\B(\rho))$,
\item $\rho$ is a Bell state.
\end{enumerate}
\end{corollary}

\begin{proof}
This equivalence follows from \cref{3.8} and \cref{2.2}.
\end{proof}

Thus, the condition $\tilde S (\tr_\A(\rho)) = \tilde S (\rho) = \tilde S(\tr_\B(\rho))$ characterizes Bell states in terms of adjusted von Neumann entropy. It also characterizes adjusted von Neumann entropy in terms of Bell states, as we now show.
For the remainder of this section, we make the finite-dimensional $C^*$-algebra for the entropy of a state explicit, writing $S(\A, \rho)$ for the von Neumann entropy of a state $\rho \in \A$, for example.

\begin{definition}\label{3.10}
A \emph{quantum entropy} is a $[0,\infty)$-valued function $H$ on pairs $(\A, \rho)$, where $\A$ is a finite-dimensional $C^*$-algebra and $\rho \in \A$ is a state, such that
\begin{enumerate}
\item $H(\A, \rho) = H (\B, \sigma)$ if there is a $\dagger$-isomorphism $\Phi\: \A \to \B$ with $\Phi(\rho) = \sigma$,
\item $H(\A \oplus \B, \rho \oplus 0) = H(\A, \rho)$,
\item $H (\A \tensor \B, \rho \tensor \sigma) = H (\A, \rho) + H(\B, \sigma)$,
\item $H(\A, p_1 \rho_1 + p_2 \rho_2) = p_1 H(\A, \rho_1) + p_2 H(\A, \rho_2) - p_1 \log p_1 - p_2 \log p_2$ if 
$p_1, p_2 \in (0,1]$ satisfy $p_1 + p_2 = 1$ and $\rho_1, \rho_2 \in \A$ satisfy $\rho_1 \perp \rho_2$.
\end{enumerate}
Each axiom implicitly assumes that both sides of the equation are defined.
\end{definition}

In the usual way, axiom 4 implies that
$$
H(\A, p_1 \rho_1 + \cdots + p_k \rho_k) = \sum_{i=1}^k p_i H(\A, \rho_i) - \sum_{i=1}^k p_i \log p_i
$$
whenever $\rho_1, \ldots, \rho_k \in \A$ are pairwise-orthogonal states and $p_1, \ldots, p_k \in (0,1]$ sum to $1$. We proceed to determine the general form of a quantum entropy, which is given in \cref{3.15}, in effect, generalizing \cref{3.3}.

\begin{proposition}\label{3.11}
Let $\A_1, \ldots, \A_s$ be finite-dimensional $C^*$-algebras, and let $\A = \A_1 \oplus \cdots \oplus \A_s$. Let $\rho = p_1 \rho_1 \oplus \cdots \oplus p_s \rho_s$ for states $\rho_i \in \A_i$ and scalars $p_i \in (0,1]$ such that $p_1 + \cdots + p_s = 1$. Then, for each quantum entropy $H$,
$$
H(\A,\rho) = \sum_{i=1}^s p_i H(\A_i, \rho_i) - \sum_{i=1}^s p_i \log p_i.
$$
\end{proposition}

\begin{proof}
We apply axioms 2 and 4 to calculate that
$$
H(\A,\rho) = \sum_{i=1}^s p_i H(\A, \hat \rho_i) - \sum_{i=1}^s p_i \log p_i =\sum_{i=1}^s p_i H(\A_i, \rho_i) - \sum_{i=1}^s p_i \log p_i,
$$
where the $i$-th component of $\hat \rho_i \in \A$ is $\rho_i$ and its other components vanish.
\end{proof}

\begin{lemma}\label{3.12}
Let $\rho \in \M_n$ be a state. Then, for each quantum entropy $H$, $ H(\M_n, \rho) = S(\M_n, \rho) + H(\M_n, E_{11}). $
\end{lemma}

\begin{proof}
Diagonalizing $\rho$, we obtain a unitary matrix $U \in \M_n$ such that $U^\dagger \rho U = p_1 E_{11} + \cdots + p_k E_{kk}$ for some $p_1, \ldots, p_k \in (0,1]$ with $k \leq n$. We now appeal to axioms 1 and 4 to calculate that
\begin{align*}
H(\M_n,\rho) & = H(\M_n, U^\dagger \rho U) = \sum_{i=1}^k p_i H(\M_{n}, E_{ii}) - \sum_{i=1}^k p_i \log p_i
\\ & = \sum_{i=1}^k p_i H(\M_{n}, E_{11}) + S(\M_n, \rho) = H(\M_n,E_{11}) + S(\M_n, \rho).
\end{align*}
\end{proof}

We will soon show that every quantum entropy $H$ is of the form
$$
H(\A, \rho) = S(\A, \rho) + \tr(\xi^F_\A \rho)
$$
for the function $F\: n \mapsto H(\M_n, E_{11})$, where $\xi_\A^F \in \A$ is defined as follows.

\begin{definition}
Let $F\: \ZZ_+ \to [0,\infty)$ be a function, and let $\A$ be a finite-dimensional $C^*$-algebra with minimal central projections $P_1, \ldots, P_s \in \A$. For each $1 \leq i \leq s$, let $n_i = \sqrt{\dim P_i \A}$. We define
$$
\xi^F_\A = F(n_1) P_1 + \cdots + F(n_s) P_s \in \A.
$$
\end{definition}

Clearly, if $\A = \M_{n_1} \oplus \cdots \oplus \M_{n_s}$, then $\xi^F_\A = F(n_1) 1_{n_1} \oplus \cdots \oplus F(n_s) 1_{n_s} \in \A$.

\begin{lemma}\label{3.14}
Let $F\: \ZZ_+ \to [0,\infty)$ be a function. For all finite-dimensional $C^*$-algebras $\A$ and $\B$,
\begin{enumerate}
\item $\Phi(\xi^F_\A) = \xi^F_\B$ if there is a $\dagger$-isomorphism $\Phi\: \A\to \B$,
\item $\xi^F_{\A \oplus \B} = \xi^F_\A \oplus \xi^F_\B$,
\item $\xi^F_{\A \tensor \B} = \xi^F_\A \tensor 1_\B + 1_\A \tensor \xi^F_\B$ if $F(n_1 n_2) = F(n_1) + F(n_2)$ for all $n_1, n_2 \in \ZZ_+$.
\end{enumerate}
\end{lemma}

\begin{proof}
Claim 1 follows from the fact that $\Phi(P)$ is a minimal central projection in $\B$ iff $P$ is a minimal central projection in $\A$ and that, in this case, $\dim \Phi(P) \B = \dim \Phi (P \A) = \dim P \A$. Claims 2 and 3 now follow from the special case in which $\A$ and $\B$ are both direct sums of full matrix algebras.
\end{proof}

\begin{theorem}\label{3.15}
A quantum entropy $H$ is uniquely determined by the numbers $H(n) \coloneqq H(\M_n, E_{11})$, for $n \geq 1$, which satisfy $H(n_1 n_2) = H(n_1) + H(n_2)$. Furthermore, for any function $F\: \ZZ_+ \to [0, \infty)$, if $F (n_1 n_2) = F(n_1) + F(n_2)$, then there exists a quantum entropy $H$ such that $H(n) = F(n)$ for all $n \geq 1$. Explicitly, $H(\A, \rho) = S(\A, \rho) + \tr(\xi^F_\A \rho)$.
\end{theorem}

\begin{proof}
Let $H_1$ and $H_2$ be quantum entropies such that $H_1(n) = H_2(n)$ for all $n \geq 1$. It follows by \cref{3.12} that $H_1(\M_n, \rho) = H_2(\M_n, \rho)$ for all $n \geq 1$ and all states $\rho \in \M_n$. It then follows by \cref{3.11} that $H_1(\A, \rho) = H_2(\A,\rho)$ for all finite-dimensional $C^*$-algebras $\A$ and all states $\rho \in \A$. Thus, $H_1 = H_2$. We apply axiom 3 to calculate that for any quantum entropy $H$, $$H(n_1 n_2) = H(\M_{n_1 n_2}, E_{11}) = H(\M_{n_1} \tensor \M_{n_2}, E_{11} \tensor E_{11}) = H(n_1) + H(n_2)$$ for all $n_1, n_2 \geq 1$, establishing the first claim of the theorem.

Let $F \: \ZZ_{+} \to [0, \infty)$ be a function such that $F (n_1 n_2) = F(n_1) + F(n_2)$ for all $n_1, n_2 \in \ZZ_+$. For each finite-dimensional $C^*$-algebra $\A$ and each state $\rho \in \A$, let $H(\A, \rho) = S(\A, \rho) + \tr(\xi^F_\A \rho)$. We prove that $H$ is a quantum entropy. First, we apply \cref{3.14} to verify that the function $(\A,\rho) \mapsto \tr(\xi^F_\A \rho)$ satisfies axioms 1, 2, and 3 of \cref{3.10}. For axiom 1, we reason that
$$
\tr(\xi^F_\A \rho) = \tr(\Phi(\xi^F_\A \rho)) = \tr(\Phi(\xi^F_\A)\Phi(\rho)) = \tr(\xi^F_\B \sigma).
$$
For axiom 2, we reason that
$$
\tr(\xi^F_{\A \oplus \B}(\rho \oplus 0)) = \tr((\xi^F_\A \oplus \xi^F_\B)(\rho \oplus 0)) = \tr(\xi^F_\A \rho \oplus 0) = \tr(\xi^F_\A \rho).
$$
For axiom 3, we reason that
$$
\tr(\xi^F_{\A \tensor \B}(\rho \tensor \sigma)) = \tr( \xi^F_\A \rho \tensor \sigma) + \tr(\rho \tensor \xi^F_\B \sigma) = \tr( \xi^F_\A \rho) + \tr( \xi^F_\B \sigma).
$$
Axioms 1, 2, and 3 are certainly also satisfied by von Neumann entropy $S$ and are preserved by addition, so $H$ satisfies these axioms as well. Similarly, the quantity $\tr(\xi^F_\A \rho)$ is clearly affine in $\rho$, so $H$ also satisfies axiom 4. Therefore, $H$ is a quantum entropy, as claimed.
\end{proof}

The basic content of \cref{3.15} is that a quantum entropy $H$ is uniquely determined by a function $F \: \ZZ_+ \to [0, \infty)$ satisfying $F(n_1 n_2) = F(n_1) + F(n_2)$. Because $F$ and $H$ have disjoint domains, we can write $H$ for the corresponding function $F$ without ambiguity. Thus, $S(n) = 0$ and $\tilde S(n) = \log n$ for all $n \in \ZZ_+$. We now introduce three properties of quantum entropies and characterize them in terms of von Neumann entropy. All three characterizations are corollaries of \cref{3.15}.

\begin{definition}\label{3.16}
Let $H$ be a quantum entropy.
\begin{enumerate}
\item It is \emph{monotone} if $H(n) \leq H(m)$ whenever $n \leq m$.
\item It is \emph{Bell} if $H(\A, \tr_\B(\rho)) = H(\A \tensor \B, \rho)$ for all Bell states $\rho \in \A \tensor \B$.
\item It is \emph{monotonic} if $H(\A, \tr_\B(\rho)) \leq H(\A \tensor \B, \rho)$ for all states $\rho \in \A \tensor \B$.
\end{enumerate}
\end{definition}

The property that $H$ is monotone is quite weak conceptually, especially in light of \cref{3.15}. It formalizes the notion that the inherent uncertainty of a fully quantum system should never decrease when the dimension increases. This property should be compared to Shannon's second axiom \cite[sec.~6]{Shannon1948}.

The property that $H$ is Bell formalizes the notion that the two parts of a bipartite system in a Bell state should contain the same information. Similarly, the property that $H$ is monotonic formalizes the notion that each part of a bipartite system in any state contains no more information than the two parts together. We emphasize that the terms ``monotone'' and ``monotonic'' denote different properties, with the former term coming from pure mathematics and the latter term coming from statistical mechanics.

We now prove three corollaries of \cref{3.15}, characterizing the three properties defined in \cref{3.16} in terms of von Neumann entropy and adjusted von Neumann entropy.

\begin{corollary}
Let $H$ be a quantum entropy. The following are equivalent:
\begin{enumerate}
\item $H$ is monotone,
\item $H(\A, \rho) = S(\A, \rho) + \gamma \tr(\rho \log \zeta_\A)$ for some constant $\gamma \geq 0$.
\end{enumerate}
The monotone entropies $S$ and $\tilde S$ correspond to $\gamma = 0$ and $\gamma = 1$, respectively.
\end{corollary}

\begin{proof}
Assume that $H$ is monotone. Thus, as a function $\ZZ_+ \to [0,\infty)$, $H$ is monotone and satisfies $H(n_1n_2) = H(n_1) + H(n_2)$ by \cref{3.15}, which implies that $H(n) = \gamma \log n$ for some $\gamma \in [0, \infty)$ \cite{Erdos1946, Khinchin1957}. For each finite-dimensional $C^*$-algebra $\A$ with minimal central projections $P_1, \ldots, P_s \in \A$, we then compute that
\begin{align*}
\xi^H_\A & = H(n_1) P_1 + \cdots + H(n_s) P_s = (\gamma \log n_1) P_1 + \cdots + (\gamma \log n_s) P_s \\ & = \gamma \log (n_1 P_1 + \cdots + n_s P_s) = \gamma \log \zeta_\A.
\end{align*}
Therefore, by \cref{3.15},
\begin{align*}
H(\A, \rho) & = S(\A, \rho) + \tr(\xi^H_\A \rho) = S(\A, \rho) + \tr(\gamma (\log \zeta_\A)\rho) \\ & = S(\A, \rho) + \gamma \tr(\rho \log \zeta_\A).
\end{align*}

Conversely, if $H(\A, \rho) = S(\A, \rho) + \gamma \tr(\rho \log \zeta_\A)$ for some $\gamma \geq 0$, then $H$ is monotone because
\begin{align*}
H(n) = H(\M_n, E_{11}) & = S(\M_n, E_{11}) + \gamma\tr(E_{11} \log \zeta_{\M_n})
\\ & = 0 + \gamma \tr(E_{11} (\log n) 1_n) = \gamma \log n.
\end{align*}
The quantum entropies $S$ and $\tilde S$ are explicitly given by $\gamma =0$ and $\gamma = 1$, respectively, and hence are monotone.
\end{proof}

\begin{corollary}\label{3.18}
Let $H$ be a quantum entropy. The following are equivalent:
\begin{enumerate}
\item $H$ is Bell,
\item $H(\A,\rho) = \tilde S (\A, \rho)$.
\end{enumerate}
\end{corollary}

\begin{proof}
Assume that $H$ is Bell. For each $n \geq 1$, let $\beta_n \in \M_n \tensor \M_n$ be the standard Bell state $ \beta_n = \frac 1 n \sum_{i, j =1}^n E_{ij} \tensor E_{ij}. $ We calculate that
\begin{align*}
2 H(n)
& =
H(\M_n, E_{11}) + H(\M_n, E_{11})
=
H(\M_n \tensor \M_n, E_{11} \tensor E_{11})
\\ & = \textstyle
H(\M_n \tensor \M_n, \beta_n)
=
H(\M_n, (\id \tensor \tr)(\beta_n))
=
H(\M_n, \frac 1 n 1_n)
\\ & = \textstyle
S(\M_n, \frac 1 n 1_n) + H(\M_n, E_{11})
= \log n + H(n),
\end{align*}
concluding that $H(n) = \log n$ for all $n$. Here, the second equality holds by axiom 3; the third equality holds by axiom 1; the fourth equality holds by \cref{3.16}(2); the sixth equality holds by \cref{3.12}. By \cref{3.15}, $H(\A,\rho) = \tilde S(\A,\rho)$ for all finite-dimensional $C^*$-algebras $\A$ and all states $\rho \in \A$. Therefore, $(1) \Rightarrow (2)$. The implication $(2) \Rightarrow (1)$ holds by \cref{3.8}.
\end{proof}

\begin{corollary}\label{3.19}
Let $H$ be a quantum entropy. The following are equivalent:
\begin{enumerate}
\item $H$ is monotonic,
\item $H(\A,\rho) \geq \tilde S (\A, \rho)$.
\end{enumerate}
\end{corollary}

\begin{proof}
Assume that $H$ is monotonic. We reason as in the proof of \cref{3.18}, using the inequality $H(\M_n \tensor \M_n, \beta_n) \geq H(\M_n, (\id \tensor \tr)(\beta_n))$ to prove that $H(n) \geq \log n = \tilde S(n)$. It follows by \cref{3.15} that $H(\A, \rho) \geq \tilde S(\A,\rho)$ for all finite-dimensional $C^*$-algebras $\A$ and states $\rho \in \A$. Therefore, $(1) \Rightarrow (2)$.

Conversely, assume that $H(\A,\rho) \geq \tilde S (\A, \rho)$. It follows that $H(n) \geq \tilde S(n)$ for all $n \geq 1$, which implies that $\xi^H_\A \geq \xi^{\tilde S}_\A = \log \zeta_\A$ for all finite-dimensional $C^*$-algebras $\A$. We now reason that, by \cref{3.4}(2) and \cref{3.14}(3),
$$
S(\A, \tr_\B(\rho)) \leq S(\A \tensor \B, \rho) + \tr (\rho (1_\A \tensor \log \zeta_\B)) \leq S(\A \tensor \B, \rho) + \tr (\rho (1_\A \tensor \xi^H_\B)),
$$
$$
S(\A, \tr_\B(\rho)) + \tr(\rho(\xi^H_\A \tensor 1_\B)) \leq S(\A \tensor \B, \rho) + \tr(\rho(\xi^H_\A \tensor 1_\B)) + \tr (\rho (1_\A \tensor \xi^H_\B)),
$$
$$
S(\A, \tr_\B(\rho)) + \tr(\tr_\B(\rho) \xi^H_\A) \leq
S(\A \tensor \B, \rho) + \tr(\rho \xi^H_{\A \tensor \B} ),
$$
$$
H(\A, \tr_\B(\rho)) \leq
H(\A \tensor \B, \rho).
$$
Therefore, $(2) \Rightarrow (1)$.
\end{proof}

Let $ \rho \in \A = \A_1 \tensor \cdots \tensor \A_n $ be a multipartite state, and for each subset $T \subseteq \{1, \ldots, n\}$, let
$$
m_\rho(T) = H\left(\bigotimes_{i \in T} \A_i, (\varphi_1 \tensor \cdots \tensor \varphi_n)(\rho)\right)
$$
where $\varphi_i = \id$ or $\varphi_i = \tr$ according to whether $i \in T$ or $i \not \in T$. Thus, $m_\rho(T)$ is the entropy of the subsystem of $\A$ that is indicated by $T$. We call $H$ \emph{localizable} iff the function $m_\rho$ is an outer measure for every multipartite state $\rho$. Shannon entropy is certainly localizable in this sense, and von Neumann entropy is certainly not.

Notionally, the information in $\A$ that is quantified by a localizable quantum entropy $H$ is present in one or more of the parts $\A_1, \ldots, \A_n$. We cannot expect $m_\rho$ to be a measure in general because the parts of this system may contain redundant information. We cannot even expect $m_\rho$ to be the image of a measure via a binary relation due to other forms of information interaction; even Shannon entropy is not alternating of order infinity \cite{Choquet1955}.

It is straightforward to show that a quantum entropy $H$ is localizable iff it is both monotonic and \emph{subadditive} in the sense that
$$
H(\A \tensor \B, \rho) \leq H(\A, \tr_\B(\rho)) + H(\B, \tr_\A(\rho))
$$
for all bipartite states $\rho \in \A \tensor \B$. We have focused on the first condition because the second condition is automatic; every quantum entropy is subadditive, as we now show.

\begin{proposition}
Let $H$ be a quantum entropy. Let $\rho \in \A \tensor \B$ be a state, where $\A$ and $\B$ are finite-dimensional $C^*$-algebras. Then,
$$H (\A \tensor \B, \rho) \leq H (\A,\tr_\B(\rho)) + H(\B,\tr_\A(\rho)).$$
Furthermore, this is an equality iff $\rho = \rho_1 \tensor \rho_2$ for states $\rho_1 \in \A$ and $\rho_2 \in \B$.
\end{proposition}

\begin{proof}
This proposition is well known in the special case of von Neumann entropy \cite[sec.~11.3.4]{NielsenChuang2010}. For the general case, we calculate that
\begin{align*}
& H(\A \tensor \B, \rho)
 =
S(\A \tensor \B, \rho) + \tr(\xi^H_{\A \tensor \B} \rho)
\\ & \leq
S(\A, \tr_\B(\rho)) + S (\B, \tr_\A(\rho)) + \tr(\xi^H_{\A \tensor \B} \rho)
\\ & =
S(\A, \tr_\B(\rho)) + S (\B, \tr_\A(\rho)) + \tr((\xi^H_\A \tensor 1_\B) \rho) + \tr((1_\A \tensor \xi^H_\B) \rho)
\\ & =
S(\A, \tr_\B(\rho)) + S (\B, \tr_\A(\rho)) + \tr(\xi^H_\A \tr_\B(\rho)) + \tr(\xi^H_\B\tr_\A(\rho))
\\ & =
 H (\A,\tr_\B(\rho)) + H(\B,\tr_\A(\rho)).
\end{align*}
The first and last equalities hold by \cref{3.15}, and the inequality holds by the special case of von Neumann entropy. We also apply \cref{3.14}(3) to the operator $\xi^H_{\A \tensor \B}$ during the computation. This computation shows that we have equality iff $S(\A \tensor \B, \rho) = S(\A, \tr_\B(\rho)) + S (\B, \tr_\A(\rho))$, which is equivalent to the condition that $\rho = \rho_1 \tensor \rho_2$ for some states $\rho_1 \in \A$ and $\rho_2 \in \B$ by the special case of von Neumann entropy.
\end{proof}

\section*{Declarations}

\noindent \textbf{Competing interests.} The authors have no competing interests to declare.

\bigskip

\noindent \textbf{Data availability.} Data sharing is not applicable to this article as datasets were neither generated nor analyzed.

\bigskip

\noindent \textbf{Funding.} The research that is reported in this article was supported by the Air Force Office of Scientific Research under Award No.~FA9550-21-1-0041 and by the National Science Foundation under Award No.~DMS-2231414.

\bibliographystyle{plain}
\bibliography{refs.bib}

\Addresses

\end{document}